\documentclass[12 pt]{article}
\usepackage{graphicx}
\usepackage{amssymb}
\usepackage[margin=0.9in]{geometry}

\usepackage[utf8]{inputenc}
\usepackage[english]{babel}
\usepackage{amsthm}
\usepackage{amsmath}
\usepackage{cleveref}
\DeclareMathOperator*{\argmax}{arg\,max}

\DeclareMathOperator*{\supp}{supp}

\usepackage{natbib}
\setcitestyle{authoryear, open={(},close={)}}

\usepackage{tikz}
\usepackage{pgfplots}
\pgfplotsset{compat=1.16}

\newtheorem{theorem}{Theorem}
\newtheorem{proposition}{Proposition}
\newtheorem{lemma}[theorem]{Lemma}
\theoremstyle{definition}
\newtheorem{definition}{Definition}[section]
\theoremstyle{remark}

\newtheorem{assumption}{Assumption}[section]

\begin{document}

\title{Using Information to Amplify Competition}
\author{Wenhao Li\thanks{I thank my advisor, Nageeb Ali, for his tremendous help on this paper. I am also grateful to Mark Armstrong, Kalyan Chatterjee, Nima Haghpanah, Yuhta Ishii, Navin Kartik, Vijay Krishna, Rohit Lamba, Xiao Lin, Vasundhara Mallick, Ron Siegel, Ece Teoman, and Kai Hao Yang for excellent and insightful comments.
        }
\thanks{Department of Economics, The Pennsylvania State University,
University Park, PA 16801 (wul260@psu.edu)}}
\maketitle
\begin{abstract}
    I characterize the consumer-optimal market segmentation in competitive markets where multiple firms selling differentiated products to consumers with unit demand. This segmentation is public---in that each firm observes the same market segments---and takes a  simple form: in each market segment, there is a dominant firm favored by all consumers in that segment. By segmenting the market, all but the dominant firm maximally compete to poach the consumer's business, setting price to equal marginal cost. Information, thus, is being used to amplify competition. This segmentation simultaneously generates an efficient allocation and delivers to each firm its minimax profit. 
   
\end{abstract}

\thispagestyle{empty} 

\newpage
\setcounter{page}{1}
\section{Introduction}
In an important paper, \cite*{bbm} (henceforth BBM) describe the limits of price discrimination. Studying the market interaction between a monopolist and a consumer, BBM characterize all possible combinations of consumer surplus and firm profits that can arise as equilibria outcomes for some information structure. A particularly important outcome is the consumer-optimal market segmentation: they show that it is possible to segment the market so that trading is efficient without improving the firm's profit beyond what it achieves without price discrimination. In this segmentation, the consumer accrues all gains from the market segmentation. 

This paper finds the consumer-optimal market segmentation in a competitive market with differentiated products. Each consumer chooses between buying the products of two or more firms and has different valuations for the product of each firm.\footnote{Some of these valuations can be lower than the value of the outside option, and so some of these consumers can be considered ``captive.''} The product differentiation may have both horizontal and vertical components. In this setting, how should one segment the market to maximize consumer surplus? 

Answering this question raises two challenges. The first is that the segmentation must be done so that each firm's equilibrium price is a best-response not only to its information but also to the other firms' prices. These market prices may be asymmetric and some market segments may lack equilibria in pure strategies. Therefore, in principle, characterizing the consumer-optimal segmentation may require accounting for the consumer's payoffs in all possible mixed strategy equilibria for all possible markets. The second challenge is that it isn't \emph{ex ante} obvious as to whether the consumer-optimal market segmentation is achieved using a public segmentation where all firms observe the same information or a private segmentation where each firm observes different information. Thus, proceeding from the consumer-optimal segmentation in a monopolistic setting to that of a competitive market requires overcoming these two obstacles. 

I characterize the consumer-optimal market segmentation and show that it is achieved using a public segmentation. I sidestep the above issues by characterizing a public segmentation that guarantees that the allocation is efficent and each firm obtains its minimax profit. This segmentation is not merely a technical device but reveals an underlying economic logic: it amplifies competition using information. 

To illustrate this logic, suppose for now that each firm has identical marginal cost.\footnote{I relax this assumption in the analysis.} In each segment of the consumer-optimal segmentation, it is commonly known which firm is the consumer's favorite. In other words, in each market segment, there is a firm $i$ so that for all consumer types in that segment, the consumer's valuation for the good of firm $i$ exceeds her valuation for that of any other firm. Thus, each market segment features only vertical differentiation. All of the non-favorite firms then realize that to poach the consumer's business, they have to charge very low prices. In equilibrium, all such non-favorite firms price at marginal cost. Their offers become the consumers' outside option if consumers choose not to purchase from firm $i$. Firm $i$ then charges its ``local monopoly price,'' assuming that outside option. At this stage, the competitive market with differentiated products reduces to the problem of a monopolist with this specific outside option. After this reduction, I invoke the uniform profit preserving extremal segmentation from BBM to find a segmentation that is allocatively efficient and in which firm $i$ obtains its minimax payoff. Hence, this is the consumer-optimal market segmentation. 

How does this affect firm profits? In BBM, the monopolist is never worse off than its payoff in the unsegmented market. This is because it can ignore the market segmentation and charge its \emph{ex ante} optimal uniform price. By contrast, in competitive markets, each firm may be worse off relative to their equilibrium profits in the unsegmented market.\footnote{The previous argument does not apply in competitive markets because even though firm $i$ could choose to ignore the segmentation, its competitors may not do so, and hence, firm $i$'s profits are no longer that of the unsegmented market.} Slicing the market into vertical components guarantees that every firm other than the favorite is competing maximally for the consumer's business, which compels the favorite firm to reduce its price. 

The paper proceeds as follows. Section 2 uses two examples to illustrates the main insights. Section 3 develops the general model and proves the main result of the paper. Section 4 concludes. The remainder of the introduction briefly discusses the most closely related prior work. 

\paragraph{Related Literature:} Several recent papers study the potential role of information design in competitive markets with production differentiation, and I discuss them below. 

The three most closely related papers are \cite{armstrong-vickers}, \cite*{bbm-note}, and \cite*{eg}. \cite{armstrong-vickers} consider a setting in which consumers non-unit demand and product differentiation is in the form of captive consumers. They show that when firms are symmetric, an unsegmented market is consumer-optimal among all public segmentations. The main difference between their analysis and mine is that I consider a general form of product differentiation with consumers that have unit demand.\footnote{In the symmetric setting studied by \cite{armstrong-vickers}, were consumers to have unit demand, we can apply the analysis of my paper. In that case, my analysis would show that the fully revealing segmentation is also consumer-optimal. My main result would then imply that both the unsegmented market and the fully segmented market are consumer-optimal among all segmentations, both public and private.} \cite*{bbm-note} study the same setting as \cite{armstrong-vickers}, but with unit demand. They show that splitting the market into ``nested'' markets where only one firm has captive consumers is producer-optimal among all public segmentations. In a different setting, \cite*{eg} characterize market segmentations that are efficient in which firms capture all of the surplus. 

The idea that information can be used to amplify competition features in \cite*{alv}. They study both monopolistic and differentiated competitive markets in which consumers voluntarily choose whether to disclose hard information about their preferences. They show that a consumer can selectively disclose information to firms to amplify competition between them. 

\cite{az} contains the novel idea that to amplify competition, consumers may choose to withhold information from themselves. They study a setting in which an initially uninformed consumer chooses how much to learn about her preferences. They find that consumers may commit to learn little so as to maintain a high degree of market competition.

\section{Two Examples}

I use two examples to illustrate the logic of my results. First, I show that the problem of Bertrand duopoly with purely vertical differentiation can be reduced to the monopolistic case considered by BBM. I then show how to segment a Bertrand duopoly market with purely horizontal differentiation into markets in which there is pure vertical differentiation. 

\subsection{Bertrand Duopoly with Vertical Differentiation}\label{Section-Vertical}
Consumers can buy a good from one of two firms, S (``superior'') and I (``inferior''). If a consumer does not buy from either firm, she obtains a payoff of zero. Half of the market comprises consumers of type $H$ and the other half is type $L$. Type $H$ has a value of $7$ for the product from firm S, and type $L$ has a value of $3$ for that good. Both types value the product from firm I at value $1$. I denote a consumer's valuation by $v\equiv (v_{S}, v_{I})$. Each firm incurs a marginal cost of $0$. The firms set prices simultaneously, and these prices are restricted to be non-negative.\footnote{Throughout this paper, I assume that firms cannot put positive probability on prices strictly below marginal cost. I describe reasons for this restriction in \Cref{Section-Model}.}

Let me first describe what happens if both firms observe the consumer's type; i.e., we segment the market into two segments, one that comprises only type $H$ consumers and the other comprises only type $L$ consumers. In each market, firms engage in Bertrand competition; hence, firm $I$ prices at $0$, firm $S$ charges the valuation difference, and consumers break ties in favor of firm $S$. The allocation is efficient, but firm $S$ reaps all of these efficiency gains. 

Relative to that benchmark, the consumer-optimal segmentation reveals partial information to both firms. This segmentation is shown in \Cref{table}: it splits the market into markets $x_{1}$ and $x_{2}$, where in market $x_{1}$, a third of consumers are type $H$, and in market $x_{2}$, all consumers are type $H$. 

\begin{table}[h]
\centering
\begin{tabular}{cccc}

Segment & $(7,1)$ &    $(3,1)$ & Mass  \\

\hline

$x_{1}$ & $\frac{1}{3}$	&	$\frac{2}{3}$  & $\frac{3}{4}$ \\[0.1cm]

$x_{2}$ & $1$	&	$0$  & $\frac{1}{4}$ \\[0.1cm]

\hline
$x_{*}$ & $\frac{1}{2}$	&	$\frac{1}{2}$  & \\[0.1cm]

\end{tabular}
\caption{Consumer-optimal Segmentation for Vertically Differentiated Market }
\label{table}
\end{table} 

What is the resulting equilibrium? Firm I always charges a price of $0$. In market $x_1$, firm S charges a price of $2$, and in market $x_2$, it charges a price of $6$. Consumers maximize their payoffs, breaking ties in favor of firm S. Firm S earns an expected payoff of $3$, firm I earns an expected payoff of $0$, and neither firm has any strictly-profitable deviations. 

This segmentation is consumer-optimal because it allocates consumers efficiently, and it leaves firms with equilibrium profits that are in fact each firm's \emph{minimax profit} for any market segmentation and equilibrium. For firm I, this is obvious. For firm S, observe that the worst-case scenario is firm I charging its lowest possible price to each consumer type, which is exactly what happens here. In that case, a feasible strategy for firm S is to ignore the segmentation and charge a price slightly below $6$, which guarantees that all high-value consumers buy from firm S. A key part of my argument is constructing these minimax profits for markets more generally and showing that they reduce to this simple case of vertically-differentiated products.

\subsection{Bertrand Duopoly with Horizontal Differentiation}

In this section, I show how to reduce a market with pure horizontal differentiation to one in which the insights from above apply. Now each consumer can purchase the product from either firm L or R. Her value from buying the good from firm $i$ is $v_i$ and we use $v\equiv (v_L,v_R)$ to denote the relevant vector of values. Suppose that $v$ is uniformly drawn from the set $\{(7,1),(3,1),(1,3),(1,7)\}$. Thus, there are extreme types who strongly prefer the product of one firm to that of the other, as well as types who have a more moderate preference. Firms again price simultaneously, have a marginal cost of $0$, and never price strictly below marginal cost. The consumer can choose to purchase the product from at most one firm, and if she decides to buy from neither firm, she obtains a payoff of zero. 

What happens in an unsegmented market? Now there exists an equilibrium in which both firms price at $7$,  more central types choose not to purchase from either firm, and extreme types buy from their favorite firms. Both firms obtain a profit of $\frac{7}{4}$ in equilibrium, and each consumer receives a payoff of zero. 

The consumer-optimal market segmentation in this setting uses four segments.  It is constructed by first slicing our horizontal market into two vertical markets: a market with type $(7,1),(3,1)$, and another market with type $(1,3),(1,7)$.  I then apply segmentation in section 2.1 to these two vertical markets because both of them are isomorphic to the market in \Cref{Section-Vertical}.  I depict this segmentation in \Cref{table2}. In market $x_{1}$, a third of consumers are of the extreme type $(7,1)$ and the rest are of the central type $(3,1)$; symmetrically, in market $x_{4}$, a third of consumers are of the extreme type $(1,7)$ and the rest are of the central type $(1,3)$.  In market $x_{2}$, all consumers are of the extreme type $(7,1)$, and symmetrically, in market $x_{3}$, all consumers are of the extreme type $(1,7)$. Notice that even though the unsegmented market has horizontal differentiation, each of these four market segments features purely vertical differentiation. In markets $x_1$ and $x_2$, it is commonly known that firm L (firm R) is the superior (inferior) firm and in markets $x_3$ and $x_4$, it is commonly known that firm R (firm L) is the superior (inferior) firm.

\begin{table}[h]
\centering
\begin{tabular}{cccccc}

Segment & $(7,1)$ & $(3,1)$  & $(1,3)$ & $(1,7)$  & Mass  \\

\hline

$x_{1}$ & $\frac{1}{3}$	&	$\frac{2}{3}$ &	$0$ &	$0$ & $\frac{3}{8}$ \\[0.1cm]

$x_{2}$ & $1$	&	$0$ &	$0$  &	$0$ & $\frac{1}{8}$ \\[0.1cm]

$x_{3}$ & $0$	&	$0$ &	$0$  &	$1$ & $\frac{1}{8}$ \\[0.1cm]

$x_{4}$ & $0$	&	$0$ &	$\frac{2}{3}$ &	$\frac{1}{3}$ & $\frac{3}{8}$ \\[0.1cm]

\hline
$x_{*}$ & $\frac{1}{4}$	&	$\frac{1}{4}$ &	$\frac{1}{4}$ &	$\frac{1}{4}$  & \\[0.1cm]

\end{tabular}
\caption{Consumer-optimal Segmentation for Horizontally Differentiated Market }
\label{table2}
\end{table}

This segmentation has the following equilibrium. In markets $x_1$ and $x_2$, firm R prices at $0$. In market $x_1$, firm L charges a price of $2$. In market $x_2$, firm L charges a price of $6$. In both of these markets, consumers purchase the product from firm L, breaking ties in its favor. Interchanging the roles of firms L and R, equilibrium behavior is symmetrical in markets $x_3$ and $x_4$. 

Observe that both firms earn an expected payoff of $\frac{3}{2}$, and by construction, have no incentive to deviate in their pricing strategy in all markets. I argue that $\frac{3}{2}$ is each firm's minimax profit: the worst scenario for each firm $i$ is the other firm charging its minimum possible price and firm $i$ having no information about the consumer's tastes. In such a case, firm $i$ can guarantee itself a profit close to $\frac{3}{2}$ by charging a price slightly below $6$, but cannot guarantee itself any profits higher than $\frac{3}{2}$. Because this segmentation allocates consumers efficiently and each firm's equilibrium profit coincides with its minimax profit, this segmentation is consumer-optimal.

\section{Model and Results}\label{Section-Model}

I study a market with $N\geq 2$ firms. Firms produce differentiated products with potentially heterogeneous marginal costs. The marginal cost of firm $i$ is $c_i\geq 0$, and I denote the vector of costs by $c\equiv (c_{1},...,c_{N})$.

There is a continuum of consumers in the unsegmented market, and the total mass is normalized to $1$. Each consumer demands at most one unit of the product, and she can obtain this product from any firm. A consumer's valuation for firm $i$'s product is $v_i$ and a consumer's type is summarized by the payoff vector $v\equiv(v_{1},...,v_{N})$. If a consumer does not purchase from any firm, she obtains a payoff of zero.  

I make two assumptions about consumer valuations. First, each consumer's valuation for the product of any firm is bounded from above by a positive number $\bar{v}$. Second, for any consumer, there is at least one firm that can serve her efficiently; i.e., for every consumer type $v$,  $\max_{i=1,..,N} \{v_{i}-c_{i}\}>0$. \footnote{This is for simplicity. If it were efficient for some consumer types to choose the outside option, then this outside option can be modeled as a ``dead'' firm who offers positive value to some consumers but its price is exogenously fixed at zero. Then  similar analysis applies.} Hence, the set of possible valuations takes the following form:
$$V \subseteq \{v\in \mathbb{R}^{N}|0<\max_{i=1,..,N} \{v_{i}-c_{i}\}\ \& \ v_{i}\leq \bar{v} \}$$. 

A market is a Borel probability measure on $V$. The set of all possible markets is defined as $\Delta(V)$.  For $ x\in \Delta(V)$, $x(B)$ is the measure of set $B\in \mathcal{B}(V)$, where $\mathcal{B}(V)$ is the collection of all Borel subset of $V$. In my analysis, I fix an original market $x_{*}\in \Delta(V)$.

To describe all possible segmentations in the market, I follow \cite{nw2020} to study a set of incomplete information games. Each firm $i$ has signal space $S_{i}$. An information structure is a joint probability measure $\psi \in \Delta(\times^{i=N}_{i=1}S_{i}\times V)$ that satisfies Bayes plausibility requirement $\psi_{v}=x_{*}$. Here, I use $\psi_{v}$ to denote the marginal distribution of $v$, and I use $\psi_{i}$ to denote marginal distribution of firm $i$'s signal.  Hence the set of possible information structure is defined as:
$$\Psi \equiv \{\psi \in \Delta(\times^{i=N}_{i=1}S_{i}\times V)|\psi_{v}=x_{*}\}.$$
For each information structure $\psi \in \Psi$, I study the following game $\Gamma(\psi)$: at stage one, $(s_{1},...,s_{N},v)$ is drawn form $\psi$, each firm $i$ privately observes $s_{i}$ and the consumer privately observes $v$; at stage two, firms simultaneously set prices; at stage three, the consumer chooses which firm's product to purchase (or not purchase any firm's product). 

In game $\Gamma(\psi)$, firm $i$ observes signals distributed as $\psi_{i}$. After observing each $s_{i}\in S_{i}$, firm $i$ hold some posterior belief about the consumer's valuation. Denote the distribution of firm $i$'s posterior believes by $\tilde{\psi}_{i}$, Bayes plausibility requires the follows:
$$\int_{x \in \Delta(V)} x(B) \tilde{\psi}_{i}(d x)=x_{*}(B), \forall B \in \mathcal{B}(V).$$
Hence each firm observes a \emph{segmentation} of original market $x_{*}$. Following BBM, the set
of segmentations of $x_{*}$ is defined as: 
\begin{equation*}
    \Sigma \equiv\left\{\sigma^{\prime} \in \Delta(\Delta(V)) | \int_{x \in \Delta(V)} x(B) \sigma^{\prime}(d x)=x_{*}(B), \forall B \in \mathcal{B}(V)\right\}
\end{equation*}

The set of public information structures $P\subset \Psi$ contains information structure $\psi$ with the following property: For each $(s_{1},...,s_{N},v)$ drawn form $\psi$, it is common knowledge that all firms hold the same posterior belief. A public information structure can be represented by a segmentation $\sigma \in \Sigma$: draw market $x\in \Delta(V)$ from $\sigma$, then draw $v$ from $x$, then privately inform each firm $x$ and the consumer $v$. 

I use the following terminology: all segmentations refer to all information structures in $\Psi$; a public segmentation refers to a information structure in $P$. If it can be represented by $\sigma$, I will call it public segmentation $\sigma$. A private segmentation refers to a information structure in $\Psi/P$.

In game $\Gamma(\psi)$, a feasible pricing rule for firm $i$ is a mapping 
$\phi_{i}: \supp \ \psi_{i} \rightarrow \Delta([c_{i},\bar{v}])$. Here, I assume firms never price strictly below their marginal cost. \footnote{If firms were allowed to charge prices below marginal cost, the consumer-optimal market segmentation is fully revealing.  With a fully revealing segmentation, the following equilibrium exists: for a consumer with value $v$, in equilibrium, firm $i$ with $i\in \argmax_{i=1,..,N} \{v_{i}-c_{i}\} $ prices at marginal cost $c_{i}$ and any other firm $j$ can set price $p_{j}$ low enough so that $v_{i}-c_{i}=v_{j}-p_{j}$. Hence, this consumer is indifferent in regard to buying from any firm. In equilibrium, she breaks ties in favor of firm $i$. Under this equilibrium, allocation is efficient, and each firm receives a profit of zero. As a result, consumer surplus is necessarily maximized.} The set of feasible pricing rules of firm $i$ in game $\Gamma(\psi)$ is denoted by $\Phi_{i}(\psi)$. 

I study perfect Bayesian equilibrium (henceforth equilibrium) of game $\Gamma(\psi)$. Notice that after firms set prices,  there are trivial sub-games where the consumer makes purchase decision. A strategy of the consumer specifies which firm's product to buy (or not buying any firm's product) in all scenarios. The set of all scenarios is $V\times [c_{i}, \bar{v}]^{N}$. Partition this set into
$E$ and $E^{c}$, where $E$ collects all scenarios with ties. Sequential rationality uniquely pins down the consumer's decision in set $E^{c}$. Hence, I only need to consider the following set of strategies for the consumer: $$T\equiv \{\tau \in (\Delta(\mathcal{N}))^{V\times [c_{i}, \bar{v}]^{N}} | \tau(v, p)(\{i^{*}\})=1, i^{*}=\argmax_{i\in \mathcal{N}} \{v_{i}-p_{i}\}, \forall (v, p)\in E^{c}\}$$ 

Here, $\mathcal{N}\equiv \{0,1,...,N\}$, $v_{0}=p_{0}=0$, $\Delta(\mathcal{N})$ is the set of all distributions on set $\mathcal{N}$. If $i$ is drawn from $\mathcal{N}$, it means the consumer choose to buy the product from firm $i$, with the convention that $0$ means not buying any firm's product. A subset of strategies for the consumer $T_{-i}$ is defined as $T_{-i}\equiv \{\tau\in T|\tau(E)\subseteq \Delta(\mathcal{N}/\{i\})\}$. Therefore, if she adopts strategy $\tilde{\tau}\in T_{-i}$, she never break ties in favor of firm $i$.

 I state the main result below.

\begin{proposition}\label{main-prop}
There exists an efficient equilibrium supported by a public segmentation with each firm receiving its minimax profit; hence, consumer surplus is maximized across equilibria supported by all segmentations. 
\end{proposition}

\begin{proof}

The proof proceeds in the following three steps:

\paragraph{Step One: Segmentation Construction.}
To construct such a public segmentation, I only need to find a $\sigma \in \Sigma$ that represents this public segmentation. For each $i\in \mathcal{N}/\{0\}$, define
\begin{equation*}
   V_{i}\equiv\{v\in V|v_{i}-c_{i}>\max_{j\in \mathcal{N}/\{i\}} \{v_{j}-c_{j}\}\},
\end{equation*}
 where  $c_{0}\equiv 0$. For consumers with valuation $v\in V_{i}$, firm $i$ is the uniquely most efficient firm for her to purchase from. Define  $V_{0}\equiv V/\cup_{i\in \mathcal{N}/\{0\}}V_{i}$. 

If for some $i\in \mathcal{N}$, I have $x_{*}(V_{i})=0$, then $V_{i}$ is a zero measure set, hence can be ignored. For the rest of the proof of this proposition, I assume $x_{*}(V_{i})>0$ for all $i\in \mathcal{N}$. Cases where  $x_{*}(V_{i})=0$ for some $i\in \mathcal{N}$ can be accommodated by only applying analysis below for $i\in \mathcal{N}$ with  $x_{*}(V_{i})>0$. 

I first slice original market $x_{*}$ into several markets each with a ``dominant'' firm and a residual market. Define market $x_{i}\in \Delta(V_{i})$ as $x_{i}(B)\equiv \frac{x_{*}(B)}{x_{*}( V_{i})}$ for $B\in \mathcal{B}(V_{i})$. For each $i\in \mathcal{N}/\{0\}$, in market $x_{i}$, firm $i$ is the ``dominant'' firm. Market $x_{0}$ is the residual market. 

Then I ``map'' each market $x_{i}$ ($i\in \mathcal{N}/\{0\}$) to a monopoly market analyzed by BBM  and then apply their uniform profit preserving extremal market segmentation. The segmentation of market $x_{i}$ will be constructed from this segmentation.

Define function $w_{i}(.)\in \mathbb{R}^{V_{i}}$ as  $w_{i}(v)\equiv v_{i}-\max_{j\in \mathcal{N}/\{i\}}\{v_{j}-c_{j}\}$ for $v\in V_{i}$. Notice that from the definition of $V_{i}$ and $w_{i}(\cdot)$, I have $w_{i}(V_{i})\subseteq[c_{i}, \bar{v}]$. 
Define mapping $T: \Delta(V_{i})\rightarrow \Delta(w_{i}(V_{i}))$ as $T(x)\equiv x\circ w^{-1}_{i}$ for $x\in \Delta(V_{i})$, where $\circ$ denotes the composite operation. Mapping $T^{-1}: \Delta(w_{i}(V_{i}))\rightarrow \Delta(V_{i})$ is defined as $T^{-1}(y)\equiv y\circ w_{i}$ for $y\in \Delta(w_{i}(V_{i}))$. It is straightforward to show that $T^{-1}$ is the inverse mapping of $T$.

Distribution $T(x_{i})$ is one-dimensional, hence it corresponds to a monopoly market: a monopoly with zero marginal cost wants to sell its product to consumers with unit demand. Each consumer's valuation is drawn from distribution $T(x_{i})$. Total mass of consumers is $x_{*}(V_{i})$. If  consumers do not buy from the monopoly, they obtain a payoff of $0$. 

It is straightforward to show that the demand of this monopoly (denoted by  $D_{i}(p)$) takes the following form:
\begin{equation*}
   D_{i}(p)= x_{*}(\{v\in V_{i}|v_{i}-\max_{j\in \mathcal{N}/\{i\}}\{v_{j}-c_{j}\}\geq p\})
\end{equation*}

Let $p^{*}_{i}$ denote one of the uniform monopoly price:
\begin{equation*}
   p^{*}_{i} \in \argmax_{p\in[c_{i}, \bar{v}]} (p-c_{i})D_{i}(p)
\end{equation*}
Notice that the maximum here is attainable because $D_{i}(p)$ is decreasing and left continuous. 

Denote the uniform monopoly profit by $\pi^{*}_{i}$, and I have:
\begin{equation}
\label{eqm-profit}
   \pi^{*}_{i}= (p^{*}_{i}-c_{i})D_{i}(p^{*}_{i})
\end{equation}

The set of segmentations of monopoly market $T(x_{i})$ is defined as: 
\begin{equation*}
    M_{i} \equiv \left\{\mu \in \Delta(\Delta(w_{i}(V_{i}))) | \int_{y \in \Delta(w_{i}(V_{i}))} y(A) \mu(d y)=T(x_{i})(A), \forall A \in \mathcal{B}(w_{i}(V_{i}))\right\}
\end{equation*}

I follow BBM to define the uniform profit preserving extremal segmentation: 
\begin{definition}
A uniform profit preserving extremal segmentation of monopoly market $T(x_{i})$ is a $\mu_{i} \in M_{i}$ such that for  each $y\in \supp \ \mu_{i} $, 
\begin{equation*}
  p^{*}_{i} \in  \supp\ y =\argmax_{p\in[c_{i}, \bar{v}]} (p-c_{i}) y(w_{i}(V_{i})\cap [p,\bar{v}]), 
\end{equation*}
\end{definition}

\begin{lemma}
There exists a uniform profit preserving extremal market segmentation of the monopoly market $T(x_{i})$. 
\end{lemma}

\begin{proof}
See Theorem 1B in BBM. 
\end{proof}

The set of segmentations of market $x_{i}$ is defined as: 
\begin{equation*}
    X_{i} \equiv \left\{\chi \in \Delta(\Delta(V_{i})) | \int_{x \in \Delta(V_{i})} x(B) \chi(d x)=x_{i}(B), \forall B \in \mathcal{B}(V_{i})\right\}
\end{equation*}

Take any uniform profit preserving extremal market segmentation $\mu_{i} \in M_{i}$, define $\chi_{i}\equiv \mu_{i} \circ T$. 
\begin{lemma}
$\chi_{i}\in X_{i}$.  
\end{lemma}
\begin{proof}
For $B \in \mathcal{B}(V_{i})$ and $A=w_{i}(B)$:
$$\int_{x \in \Delta(V_{i})} x(B) \chi_{i}(d x)=\int_{x \in \Delta(V_{i})} x\circ w^{-1}_{i}(A) \chi_{i}(d x)=\int_{x \in \Delta(V_{i})} T(x)(A) \chi_{i}(d x)=$$
$$\int_{y \in \Delta(w_{i}(V_{i}))} T\circ T^{-1}(y)(A) \mu_{i}(d y)=\int_{y \in \Delta(w_{i}(V_{i}))} y(A) \mu_{i}(d y)=$$
$$T(x_{i})(A)=T(x_{i})\circ w_{i}(B)=T^{-1}\circ T(x_{i})(B)=x_{i}(B)$$
The second equality is by definition of $T$. The third equality is by changing of variable and viewing $\chi_{i}$ as a push-forward for $\mu_{i}$. The fifth equality is by $\mu_{i} \in M_{i}$. The seventh equality is by definition of $T^{-1}$. 
\end{proof}

I segment market $x_{i}$ by $\chi_{i}$.  Define $\chi_{0}\in \Delta(\Delta(V_{0}))$ as $\chi_{0}(S)\equiv \mathbb{I}_{x_{0}\in S}$ for $S\in \mathcal{B}(\Delta(V_{0}))$. Notice that trivially, $\chi_{0}$ is a segmentation of $x_{0}$. For $i\in \mathcal{N}$,  any $x\in \supp\ \chi_{i}$ is only defined on $\mathcal{B}(V_{i})$. I extend its definition to $\mathcal{B}(V)$ by defining $x(B)\equiv x(B\cap V_{i})$ for $B\in \mathcal{B}(V)$. Define $\sigma_{*}$ as $\sigma_{*}(S)\equiv \Sigma_{i\in \mathcal{N}} x_{*}(V_{i}) \chi_{i}(S\cap \Delta(V_{i}))$ for $S\in \mathcal{B}(\Delta(V))$.

\begin{lemma}
$\sigma_{*}\in \Sigma$. 
\end{lemma}
\begin{proof}
Take any $B\in \mathcal{B}(V)$, I have:
$$\int_{x \in \Delta(V)} x(B) \sigma_{*}(d x)=\Sigma_{i\in \mathcal{N}} x_{*}(V_{i}) \int_{x \in \Delta(V_{i})}x(B) \chi_{i}(d x)=\Sigma_{i\in \mathcal{N}} x_{*}(V_{i})\int_{x \in \Delta(V_{i})} x(B\cap V_{i}) \chi_{i}(d x)=$$
$$\Sigma_{i\in \mathcal{N}} x_{*}(V_{i})x_{i}(B\cap V_{i})=\Sigma_{i\in \mathcal{N}} x_{*}(B\cap V_{i})=x_{*}(B)$$

The first equality is by definition of $\sigma_{*}$. The third equality is by $\chi_{i}\in X_{i}$. The fourth equality is by definition of $x_{i}$. 
\end{proof} 

Segmentation $\sigma_{*}$ is constructed by first segmenting original market $x_{*}$ into each market $x_{i}$, and then segmenting each market $x_{i}$ further by market segmentation $\chi_{i}$. I will use public segmentation $\sigma_{*}$ to prove this proposition. To construct an equilibrium under this public segmentation, it is enough to construct an equilibrium in each market segment of $\sigma_{*}$. 

\paragraph{Step Two: Equilibrium Construction.}

\begin{lemma}
Under public segmentation $\sigma_{*}$, the following strategy profile constitutes an equilibrium in market $x_{0}$: each firm charges its marginal cost and the consumer uses any strategy $\tau\in T$. 
\end{lemma}

\begin{proof}
Notice that in market $x_{0}$, the consumer's valuation $v$ is drawn from set $V_{0}$. For any firm $i$, given equilibrium strategies of opponent firms, if firm $i$ charges any price $p>c_{i}$, no type $v$ of consumer will buy its product because $v_{i}-p<v_{i}-c_{i}\leq v_{j}-c_{j}$, for some firm $j$. Such a firm $j$ exists because $v\in V_{0}$ implies that there are at least two firms, $i$ and $j$ with $i, j\in \argmax_{k\in \mathcal{N}}\{v_{k}-c_{k}\}$. 
\end{proof}

Recall that under public segmentation $\sigma_{*}$, each market $x_{i}$ ($i\in \mathcal{N}/\{0\}$) is further segmented by $\chi_{i}$. 

\begin{lemma}
Under public segmentation $\sigma_{*}$, the following strategy profile constitutes an equilibrium in market $x_{i}$ ($i\in \mathcal{N}/\{0\}$):  
For each $x \in \supp\ \chi_{i}$, firm $i$ charges $\min \supp\ T(x)$. All other firms charge every consumer their marginal cost. The consumer makes purchase decision to maximize her payoff, and in case of ties, she breaks ties in favor of firm $i$. 

Under this equilibrium, firm $i$ has equilibrium profit $\pi^{*}_{i}$; all types of consumer buy from firm $i$; hence, other firms have zero profit. 
\end{lemma}
\begin{proof}
In such market $x_{i}$, the consumer has valuation $v$ drawn from set $V_{i}$. A consumer with valuation $v$ has equilibrium willingness to pay for firm $i$ equaling $w_{i}(v)$.  By construction of $\chi_{i}$, I have that for each  $x\in \supp\ \chi_{i}$, $T(x) \in \supp\ \mu_{i}$. In market segment $x$, the consumer's equilibrium willingness to pay for firm $i$ is distributed as $T(x)$. By $T(x) \in \supp\ \mu_{i}$ and that $\mu_{i}$ is uniform profit preserving extremal market segmentation, it is optimal for firm $i$ to charge $\min \supp\ T(x)$. Also, it is equally optimal for firm $i$ to charge $p^{*}_{i}$. Because firm $i$ charges lowest willingness to pay for itself and the consumer breaks ties in favor of firm $i$, all types of consumer in segment $x$ buy from firm $i$. Therefore, no types of consumer buy from any other firm $j$ in equilibrium in segment $x$. That is, given equilibrium strategies of other players, even if firm $j$ charge its marginal cost, no types of consumer buy its product. Hence, firm $j$ has no profitable deviation. 

Because in equilibrium, in each segment, it is equally optimal for firm $i$ to charge $p^{*}_{i}$, firm $i$ obtains the same profit as if it charges all consumers in market $x_{i}$ uniform price $p^{*}_{i}$, resulting in a profit of $\pi^{*}_{i}$. 
\end{proof}

Notice that for each firm $i\in \mathcal{N}/\{0\}$, if I aggregate its equilibrium profits across all of the markets, I will have that firm $i$ obtains a total profit of $\pi^{*}_{i}$ in equilibrium. 

\paragraph{Step Three: Equilibrium Profits are Minimax Profits.}
The minimax profit of firm $i$ across equilibria and all segmentations is defined as:
\begin{equation*}
   \pi^{m}_{i}\equiv \inf_{\{\psi\in \Psi,\phi_{-i}\in \Phi_{-i}(\psi),  \tau \in T\}} \sup_{\phi_{i}\in \Phi_{i}(\psi)} \pi_{i}(\phi_{i},\phi_{-i}, \tau)
\end{equation*}
where $\Phi_{-i}(\psi)\equiv \times_{j\in \mathcal{N}/\{0,i\}} \Phi_{j}(\psi)$ and  $\pi_{i}(\phi_{i},\phi_{-i}, \tau)$ is firm $i$'s profit under strategy profile $(\phi_{i},\phi_{-i}, \tau)$. 

\begin{lemma}\label{L2}
 In the above equilibrium supported by public segmentation $\sigma_{*}$, each firm's equilibrium profit equals to its minimax profit. 
\end{lemma}

\begin{proof}
Since $\pi^{*}_{i}$ is firm $i$'s equilibrium profit under public segmentation $\sigma_{*}$, by definition of $\pi^{m}_{i}$, we have $\pi^{*}_{i}\geq \pi^{m}_{i}$. To finish proving this Lemma, it is enough to show $\pi^{m}_{i}\geq \pi^{*}_{i}$. If $\pi^{*}_{i}=0$, this is obviously true. If $\pi^{*}_{i}>0$, then by equation (\ref{eqm-profit}), I will have $p^{*}_{i}>c_{i}$. 

For arbitrary $\psi\in \Psi, \phi_{-i}\in \Phi_{-i}(\psi), \tau \in T, \tilde{\tau}\in T_{-i}$, and $p \in [c_{i}, \bar{v}]$, 
$$\sup_{\phi_{i}\in \Phi_{i}(\psi)} \pi_{i}(\phi_{i},\phi_{-i}, \tau)\geq \pi_{i}(p,\phi_{-i}, \tau)\geq \pi_{i}(p,\phi_{-i}, \tilde{\tau})\geq \pi_{i}(p,c_{-i}, \tilde{\tau})=(p-c_{i})w_{i}(p)$$
 
 Where $w_{i}(p)$ is firm $i$'s demand under strategy profile $(p,c_{-i}, \tilde{\tau})$, and it is straightforward to show $w_{i}(p)= x_{*}(\{v\in V|v_{i}-\max_{j\in \mathcal{N}/\{i\}}\{v_{j}-c_{j}\}> p\})$.
 
The first inequality follows from $p \in \Phi_{i}(\psi)$. Here $p$ as a strategy of firm $i$ is interpreted as the pricing rule that charges uniform price $p$ after any signal.  

The second inequality is true because  $\tau$ may specify the consumer to break ties in favor of firm $i$, while $\tilde{\tau}$ never does so. 

To show the third inequality: notice that under strategy profile $(p,\phi_{-i}, \tilde{\tau})$ and $(p,c_{-i}, \tilde{\tau})$, firm $i$ prices at the same level $p$. Hence, it is sufficient to show that for $v\in V$, $j\in \mathcal{N}/\{0,i\}$, $p_{j}$ drawn from $\phi_{j}(s)$ and $s$ drawn from $\psi_{j}$, 
$$\tilde{\tau}(v, p,c_{-i})(\{i\})>0\implies  \tilde{\tau}(v, p,p_{-i})(\{i\})=1.$$
This is true because:
$$\tilde{\tau}(v, p,c_{-i})(\{i\})>0 \implies v_{i}-p>\max_{j\in \mathcal{N}/\{i\}}\{v_{j}-c_{j}\}\implies v_{i}-p>v_{j}-p_{j}.$$
The last step is by $p_{j}\geq c_{j}$. 

Thus I have $\sup_{\phi_{i}\in \Phi_{i}(\psi)} \pi_{i}(\phi_{i},\phi_{-i}, \tau)\geq (p-c_{i})w_{i}(p)$. Taking supremum over $p\in [c_{i}, \bar{v}]$ on the right-hand-side and then taking infimum over $\psi \in \Psi,\phi_{-i}\in \Phi_{-i}(\psi),  \tau\in T$ on the left-hand-side, I have $\pi^{m}_{i}\geq \sup_{p\in[c_{i}, \bar{v}]} (p-c_{i})w_{i}(p)$.

Now it is enough to show $\sup_{p\in[c_{i}, \bar{v}]} (p-c_{i})w_{i}(p)\geq \pi^{*}_{i}$. For $p\geq c_{i}$, I have $w_{i}(p)= x_{*}(\{v\in V_{i}|v_{i}-\max_{j\in \mathcal{N}/\{i\}}\{v_{j}-c_{j}\}> p\})$. Then $p^{*}_{i}>c_{i}$ will imply $\lim_{p \uparrow p^{*}_{i}} w_{i}(p)=D_{i}(p^{*}_{i})$.  This further imply $\sup_{p\in[c_{i}, \bar{v}]}  (p-c_{i})w_{i}(p)\geq (p^{*}_{i}-c_{i})D_{i}(p^{*}_{i}) =\pi^{*}_{i}$. 
\end{proof}
Step 3 completes the proof of \Cref{main-prop}. Because in each market segment, every type of consumer buys from the most efficient firm for her, the allocation is efficient. Yet, each firm's equilibrium profit coincides with its minimax profit. Therefore,  public segmentation $\sigma_{*}$ is consumer-optimal.
\end{proof}

My second result compares firms' profit under consumer-optimal segmentation to their profit under the unsegmented market. I make the following assumption to facilitate this comparison:

\begin{assumption}\label{assump}
Surplus from trade between any consumer firm pair is bounded below by a positive number $\underline{u}$, that is, $v_{i}-c_{i}\geq \underline{u}$, for all $i\in \mathcal{N}/\{0\}$ and for all  $v\in V$. 
\end{assumption}

Recall that under the consumer-optimal segmentation, in each market segment, dominated firms price at their marginal cost in equilibrium. Assumption \ref{assump} would imply that the consumer could obtain a positive payoff by purchasing from any dominated firm. Hence, she never considers not buying from any firm. By ruling out this possibility of not trading, it is easier to compare a firm's profit under these two scenarios.

\begin{proposition}
Compared to any equilibrium under the unsegmented market with positive profit for each firm, all firms are strictly worse off under the equilibrium supported by the consumer-optimal segmentation.\footnote{The requirement that all firms obtain positive profit under the unsegmented market can be satisfied if $x_{*}(V_{i})>0$ for any $i\in \mathcal{N}/\{0\}$. Lemma 6 implies that equilibrium profits cannot be lower than $\pi_{i}^{*}$, and it is straightforward to show that $\pi_{i}^{*}>0$ if $x_{*}(V_{i})>0$. } 
\end{proposition}

\begin{proof}
For all $i\in \mathcal{N}/\{0\}$, let $(\bar{\phi},\bar{\tau})$ be any equilibrium under the unsegmented market with equilibrium profit $\bar{\pi}_{i}>0$. Hence $\bar{\pi}_{i}=\pi_{i}(\bar{\phi},\bar{\tau})$.  I need to show $\bar{\pi}_{i}>\pi^{*}_{i}$. 

If $\pi^{*}_{i}=0$, then I immediately have $\bar{\pi}_{i}>\pi^{*}_{i}$. 
If $\pi^{*}_{i}>0$, then by equation (\ref{eqm-profit}), I have $p^{*}_{i}>c_{i}$ and $D_{i}(p^{*}_{i})>0$. 

For $i\in \mathcal{N}/\{0\}$, $supp\ \bar{\phi}_{i}\subseteq [\bar{\pi}_{i}+c_{i}, \bar{v}]$.  Because, $\bar{\pi}_{i}=\pi_{i}(p,\bar{\phi}_{-i},\bar{\tau})\leq (p-c_{i}), \forall p\in \supp\ \bar{\phi}_{i}$. 

For any $p\in [c_{i}, p^{*}_{i})$, $\tilde{\tau}\in T_{-i}$ and $\epsilon \triangleq \min \{\frac{\underline{u}}{2}, \bar{\pi}_{1},...,\bar{\pi}_{N}\}$, I have: 
$$\pi_{i}(\bar{\phi}_{i}, \bar{\phi}_{-i},\bar{\tau})\geq \pi_{i}(p+\epsilon, \bar{\phi}_{-i},\bar{\tau}) \geq \pi_{i}(p+\epsilon, \bar{\phi}_{-i},\tilde{\tau})\geq \pi_{i}(p+\epsilon, \bar{\pi}_{-i}+c_{-i},\tilde{\tau}).$$ Here, vector $\bar{\pi}_{-i}+c_{-i}$ denotes that each firm $j$ other than firm $i$ prices at $\bar{\pi}_{j}+c_{j}$. These three inequalities follow the same logic as the proof of the three corresponding inequalities in the proof of Lemma 6. 

My proof of this proposition will be done if I have $x_{*}(v\in V|v_{i}-p-\epsilon>\max_{j\in \mathcal{N}/\{i\}} \{v_{j}-\bar{\pi}_{j}-c_{j}\})\geq w_{i}(p)$. Because:
$$\pi_{i}(\bar{\phi}_{i}, \bar{\phi}_{-i},\bar{\tau})\geq \lim_{p \uparrow p^{*}_{i}}\pi_{i} (p+\epsilon, \bar{\pi}_{-i}+c_{-i},\tilde{\tau})=(p^{*}_{i}+\epsilon-c_{i}) \lim_{p \uparrow p^{*}_{i}} x_{*}(v\in V|v_{i}-p-\epsilon>\max_{j\in \mathcal{N}/\{i\}} \{v_{j}-\bar{\pi}_{j}-c_{j}\})$$
$$>(p^{*}_{i}-c_{i}) \lim_{p \uparrow p^{*}_{i}}w_{i}(p)=(p^{*}_{i}-c_{i})D_{i}(p^{*}_{i})=\pi^{*}_{i}.$$
Here, the second step is by $\pi_{i}(p+\epsilon,\bar{\pi}_{-i}+c_{-i},\tilde{\tau})=(p+\epsilon-c_{i})x_{*}(v\in V|v_{i}-p-\epsilon>\max_{j\in \mathcal{N}/\{i\}} \{v_{j}-\bar{\pi}_{j}-c_{j}\})$. Notice that $x_{*}(v\in V|v_{i}-p-\epsilon>\max_{j\in \mathcal{N}/\{i\}} \{v_{j}-\bar{\pi}_{j}-c_{j}\})$ and $w_{i}(p)$ have left limits because they are both decreasing in $p$. The third and the fourth steps are by $\lim_{p \uparrow p^{*}_{i}} w_{i}(p)=D_{i}(p^{*}_{i})$, $p^{*}_{i}>c_{i}$ and $D_{i}(p^{*}_{i})>0$.

To show $x_{*}(v\in V|v_{i}-p-\epsilon>\max_{j\in \mathcal{N}/\{i\}} \{v_{j}-\bar{\pi}_{j}-c_{j}\})\geq w_{i}(p)$: notice that left-hand-side is firm $i$'s demand under strategy profile $(p+\epsilon,\bar{\pi}_{-i}+c_{-i},\tilde{\tau})$ and right-hand-side is firm $i$'s demand under strategy profile $(p,c_{-i},\tilde{\tau})$. Therefore, it is sufficient to show that for $v\in V$: $$\tilde{\tau}(v, p,c_{-i})(\{i\})>0\implies  \tilde{\tau}(v, p+\epsilon,\bar{\pi}_{-i}+c_{-i})(\{i\})=1.$$
 This is true because for $j\in \mathcal{N}/\{0,i\}$: 
 $$\tilde{\tau}(v, p,c_{-i})(\{i\})>0\implies   v_{i}-p>\max_{j\in \mathcal{N}/\{i\}} \{v_{j}-c_{j}\} \implies v_{i}-p-\epsilon>v_{j}-c_{j}-\epsilon\geq \underline{u}-\epsilon\geq \frac{\underline{u}}{2}>0\ \&$$
$$v_{i}-p-\epsilon >v_{j}-c_{j}-\epsilon\geq v_{j}-\bar{\pi}_{j}-c_{j}. $$

\end{proof}

\section{Conclusion}

This paper has shown how oligopolistic markets with product differentiation can be segmented to amplify competition. By slicing the market into vertical components, firms compete in several vertical markets, each with a dominating firm favored by all consumers in that vertical market. Realizing the dominant firm's advantage, dominated firms compete maximally for consumers' business by pricing at their marginal cost. Thus, the dominant firm knows that other firms are trying to poach its consumers. To induce this dominating firm to charge low prices, I apply the uniformly profit preserving extremal segmentation from BBM  to each vertical market. In this way, information generates an equilibrium in which each firm's profit is its minimax profit across all equilibria and segmentations, and nevertheless, the equilibrium allocation is efficient. Thus, this public segmentation is consumer-optimal among all segmentations. 

\bibliography{re}

\end{document}